\documentclass[10pt, conference, compsocconf]{IEEEtran}
\ifCLASSINFOpdf
\else
\fi

\usepackage{epsfig}
\usepackage{amssymb}
\usepackage{graphicx}
\usepackage{amsmath}
\usepackage{amsthm}
\usepackage{subfig}
\usepackage{multirow}
\usepackage{times}
\usepackage{graphicx}
\usepackage{setspace}
\usepackage{graphicx}
\usepackage{caption}
\usepackage{subfloat}
\RequirePackage{epsfig}
\usepackage{setspace}
\usepackage{cite}
\usepackage{amssymb}
\usepackage{verbatim} 
\usepackage{algorithm}
\usepackage{listings}
\usepackage{textcomp}
\usepackage{relsize}
\usepackage{epstopdf}

\newcommand{\argmax}{\operatornamewithlimits{argmax}}

\newcommand{\ud}{\mathrm{d}}
\newtheorem{theorem}{Theorem}

\newtheorem{proposition}{Proposition}
\newtheorem{remark}{Remark}

\usepackage{array}
\newcolumntype{H}{>{\setbox0=\hbox\bgroup}c<{\egroup}@{}}

\begin{document}
%
\title{Generative Maximum Entropy Learning for Multiclass Classification}
\author{\IEEEauthorblockN{Ambedkar Dukkipati, Gaurav Pandey, Debarghya Ghoshdastidar, Paramita Koley, D. M. V. Satya Sriram}
\IEEEauthorblockA{Department of Computer Science and Automation\\
        Indian Institute of Science \\
        Bangalore, 560012, India\\
Email:\{ad, gp88, debarghya.g, paramita2000, dmvsriram\}@csa.iisc.ernet.in}}

\maketitle

\begin{abstract}
Maximum entropy approach to classification is very well studied in
applied statistics and machine learning and almost all the methods
that exists in literature are discriminative in nature. In this paper,
we introduce a maximum entropy classification method with feature
selection for large dimensional 
data such as text datasets that is generative in nature. 
To tackle the curse of dimensionality of large data sets,
we employ conditional independence assumption
(Naive Bayes) and we perform feature selection simultaneously, by
enforcing a `maximum  discrimination' between estimated class
conditional densities. For two class problems, in the proposed method,
we use Jeffreys ($J$) divergence to discriminate the class conditional
densities. To extend our method to the multi-class case, we propose a
completely new approach by considering a multi-distribution
divergence: we replace Jeffreys divergence by Jensen-Shannon ($JS$)
divergence  to discriminate conditional densities of multiple
classes. In order to reduce computational complexity, we employ a
modified Jensen-Shannon divergence ($JS_{GM}$), based on AM-GM
inequality. We
show that the resulting divergence is a natural  
generalization of Jeffreys divergence to a multiple distributions
case.
As far as the theoretical justifications are concerned
we show that when one intends to select the best features in a
generative maximum entropy approach, maximum discrimination using
$J-$divergence emerges naturally in binary classification.
Performance and comparative study of the proposed algorithms have been
demonstrated on large dimensional text and gene expression datasets that show our
methods scale up very well with large dimensional datasets.
\end{abstract}

\IEEEpeerreviewmaketitle
\section{Introduction}
Broadly, supervised learning can be divided into discriminative
learning and generative
learning~\cite{NgJordan:2002:OnDiscriminativeVsGenerativeClassifiers}. In
the generative approach to classification the aim is to model the   
joint distribution of the data and the class labels  from the
training data. One can then compute the class conditional densities
for each class and then assign  the instance to the class with
highest posterior probability. Examples of generative classifiers include
linear discriminant analysis (LDA) and Bayes
classifier.

On the other hand, one can directly model a discriminating function
without actually constructing a model for the data. A discriminating
function may be chosen so as to minimize some measure of error on the
training data. Such an approach is termed as a discriminative
classification. Statistical discriminative classification, a subclass
of discriminative   classification, models the posterior density
directly, which is then used to classify a new instance. Examples of   
discriminative classifiers include logistic Regression and
support vector machines (SVM).

The generative classifiers have a smaller variance 
than their discriminative counterparts, and hence require lesser data
for training~\cite{NgJordan:2002:OnDiscriminativeVsGenerativeClassifiers}
to achieve their asymptotic error. This is in contrast to discriminative models, 
which tend to overfit the data, when the number of training instances is small.
Furthermore, incomplete data and latent variables can be taken care of in
generative models. Moreover, since generative models model the generation 
of the entire data (hidden as well as observed), one can incorporate 
complex dependencies between data/features in the model, thereby allowing
the construction of models that are closer to the true data generating mechanism.
Any domain knowledge about the data-generating mechanism can be incorporated in
a generative model quite easily. Lastly, generative models are 
more intuitive to understand than their discriminative counterparts. 

After the works
of~\cite{Jaynes:1957:InformationTheoryStatisticalMechanics_I,Kullback:1959:InformationTheoryAndStatistics},   
a variety of statistical methods and machine learning techniques have 
taken up ideas from information theory
\cite{ZhuWuMumford:1997:MinimaxEntropyPrincipleAndItsApplication,DudikPhillipsSchapire:2007:MaximumEntropyDensityEstimation}.     
Maximum entropy or minimum divergence methods form a subclass of such
techniques where the aim is to make minimum  
assumption about the data. These approaches proved to be more useful
in the field 
of natural language processing and text
classification~\cite{BeefermanBergerLafferty:1999:StatisticalModelsForTextSegmentation,NigamLaffertyMcCallum:1999:UsingMaximumEntropyForTextClassification},    
where the curse of dimensionality becomes more significant. Variants 
of maximum entropy techniques have also been considered  
in literature. For instance,
regularized maximum
entropy models have been considered
in~\cite{DudikPhillipsSchapire:2007:MaximumEntropyDensityEstimation}.

In spite of the vast literature surrounding maximum entropy models,
almost all classification methods 
considered are discriminative in nature. The reason for this is that
the partition function for most generative models  
cannot be obtained in closed form. However, as mentioned earlier, when
the number of training instances is small,
a discriminative model tends to overfit the data, and hence, a
generative model must be preferred. Hence, we explore  
generative maximum entropy models in this paper and compare it with
other discriminative methods. The use of a generative model  
also allows us to incorporate feature selection simultaneously.

\subsection*{Contributions}
We propose a new method of
classification using a generative model to estimate the class
conditional densities.  Furthermore, we perform feature selection using a
discriminative criteria based on Jeffreys divergence. We
call this method as Maximum Entropy with Maximum Discrimination~(MeMd).   
The basic approach is based on the idea presented
in~\cite{DukkipatiYadavMurty:2010:MaximumEntropyModelBasedClassification}
that does not scale up for large datasets, and further, extension to
multi-class classification is not studied. We improve on the  time
complexity of the algorithm  in 
\cite{DukkipatiYadavMurty:2010:MaximumEntropyModelBasedClassification}
by assuming class  conditional independence of all the features. 

Most of the classification methods that are designed for binary case
can be extended  to multi-class case by formulating the problem as several 
binary classification
problems~\cite{AllweinShapireSinger:2000:ReducingMulticlassToBinary}. Among
these `one-vs-all' is well known and successfully applied to
SVMs~\cite{RifkinKlautau:2004:InDefenceOfOneVsAllClassification}. 
 
One of our main contributions in this paper is the unique way we
extend our binary classification method (MeMd) to multi-class case.  
The main idea is to use Jensen-Shannon divergence, which can be
naturally defined for more than two distributions  
(these divergences are known as multi-distribution
divergences~\cite{GarciaWilliamson:2012:DivergencesAndRisksForMulticlassExperiments}).  
To simplify the calculations, as well as the computational complexity,
we replace arithmetic mean in $JS-$divergences with geometric mean and
study performance of MeMd. We show that this leads to a
multi-distribution extension of Jeffreys
divergence.

We perform experimental study of the proposed method
on some large benchmark text datasets and show that the proposed method  
is able to do  drastic dimensionality reduction and also give good
accuracies comparable to SVMs and outperforms discriminative approaches.

\section{Preliminaries And Background}
\label{sec_preliminaries}
We dedicate this portion to set up the notations. The problem at hand
is that of classification, \textit{i.e.}, we are given a training data
set with class labels {$\{c_1, c_2, \ldots, c_M\}$}.  
In binary classification, we have $M=2$.
Consider each instance of the data to be 
of the form {$\mathbf{x} = (x_1, x_2, \ldots, x_d)$}, where {$x_i \in \mathcal{X}_i$} 
is the {$i^{th}$} feature of the data. So the input space is 
{$\mathcal{X} = \mathcal{X}_1 \times \mathcal{X}_2 \times \ldots
  \times \mathcal{X}_d$}. In this setting, our objective is to rank
the features according to their `discriminative' ability.  

A key step in the proposed algorithms is estimation of the class conditional densities,
denoted as {$P_{c_j}(.) := P(.|c_j)$}, for each each class. This step requires some statistics of 
the data in the form of expected values of certain functions 
{$\Gamma = \{\phi_1(\mathbf{x}), \phi_2(\mathbf{x}), \ldots, \phi_l(\mathbf{x})\}$},
where {$\phi_j, j=1,\ldots,l$}, is defined over the input space {$\mathcal{X}$}. 
These are often termed as \textit{feature
  functions}~\cite{DudikPhillipsSchapire:2007:MaximumEntropyDensityEstimation},
and should not be confused with the features.  
In fact, the feature functions can be chosen so that they result in
the moments of the individual  features, for example it can be of the
form {$\phi(\mathbf{x}) = x_i^k$}, whose expected value gives the
{$k^{th}$} moment of the {$i^{th}$} feature. 

Let $\mathbf{X} = (X_1, \ldots, X_d)$ be a random vector that takes values
from the set {$\mathcal{X}$}.
Suppose the only information (observations) available about the distribution
of $\mathbf{X}$ is in the form of expected values 
of the real-valued feature functions $\Gamma = \{ \phi_{1}(\mathbf{x}), \ldots, \phi_{l}(\mathbf{x})\}$. 
We therefore have, 
\begin{equation}\label{eq2.2.2}
 \mathsf{E}_{P} \left[ \phi_{r}(\mathbf{X})\right] =
 \int_{\mathcal{X}} \phi_r(\mathbf{x}) P(\mathbf{x}) \:
 \ud \mathbf{x} 
 \enspace,\:\:\: 
 r=1,2, \ldots, l,
\end{equation}  
where these expected values are assumed to be known. 
In the maximum
entropy approach to density estimation one choose the distribution of
the form 
\begin{equation}
\label{eq2.2.4}
P(\mathbf{x}) = \exp\left( -\lambda_{0}-\sum_{j=1}^{l}\lambda_j 
	  \phi_j(\mathbf{x})\right) \enspace, 
\end{equation}
      where $  \lambda_{0},\lambda_{1}, \ldots ,\lambda_{l}$ 
      are obtained by
      replacing~\eqref{eq2.2.4} in~\eqref{eq2.2.2} and
      $\int_{\mathcal{X}} P(\mathbf{x}) d\mathbf{x} = 1$. 
      This is known as the maximum entropy (ME) distribution.

In presence of an observed data $\{\mathbf{x}^{(k)}, k=1,\ldots, N\}$, 
       in the maximum entropy modeling one assumes  that 
       the expected values of the moment constraint functions
       \textit{i.e.}, 
       $\mathsf{E}_P\left[\phi_r(\mathbf{x})\right]$ can be
       approximated by observed statistics or sample means 
      of $\phi_r(\mathbf{x})$, $r=1,\ldots,
      l$~\cite{NigamLaffertyMcCallum:1999:UsingMaximumEntropyForTextClassification}. Therefore
      we set   
      \begin{equation}
	\mathsf{E}_P \left[ \phi_r (\mathbf{x})\right]  \approx
        \frac{1}{N}\sum_{k=1}^{N} \phi_r (\mathbf{x}^{(k)}) =
        \mu_{r}^{\mathrm{emp}},\:\:\: r=1,\ldots, l.
      \end{equation}

For the estimation of $\Lambda = (\lambda_0, \ldots, \lambda_l)$ one can
show that the maximum likelihood estimator
       $\Lambda^{'}$ is given by, 
       \begin{equation}\label{likelihood2}
 	\Lambda^{'} = \argmax_{\Lambda} \frac{1}{N}
 	\sum_{i=1}^{N}\ \ln P(\mathbf{x}^{(i)}; \Lambda, \Gamma)\enspace,
       \end{equation}
       where, $\mathbf{x}^{(1)}, \ldots,\mathbf{x}^{(N)}$ are assumed to be $i.i.d.$ samples
       drawn from a unknown distribution. To solve~\eqref{likelihood2}, one can use iterative methods like gradient
       descent method or iterative scaling methods
       \cite{DarrochRatcliff:1972:IterativeScaling,Csiszar:1989:GeometricInterpretation}.

Moving on to divergences,              
a symmetrized version of
$KL-$divergence~\cite{Kullback:1959:InformationTheoryAndStatistics}
is known as Jeffreys 
divergence~\cite{Jeffreys:1946:AnInvariantFormForThePriorProbability},
or simply $J-$divergence. Given two pdfs $P$ and $Q$, $J-$divergence
between them is defined as  
\begin{align}
\label{eq:jef}
J(P \parallel Q) &=  KL(P \parallel Q) + KL(Q \parallel P)	\nonumber
        \\&=  \int_{\mathcal{X}}( P(x)-Q(x) ) \ln \frac{P(x)}{Q(x)}
        \: \ud {x} \enspace.
\end{align}

Although $J-$divergence was proposed in the context of statistical
estimation problems to  provide measures of the discrepancy between two
  laws~\cite{Jeffreys:1946:AnInvariantFormForThePriorProbability}, 
its connection to $KL-$divergence has made it popular in  
classification
tasks~\cite{NishiiEguchi:2006:ImageClassificationJeffreysDivergence,DeselaersKeysersNey:2004:FeaturesForImageRetrieval}. 
Its relationship with other divergence measures have been studied
in~\cite{DragomirSundeBuse:2000:InequalitiesforJeffreys}. 
In fact, this $J-$divergence can also be obtained as a $f-$divergence~\cite{AliSilvey:1966:AgeneralClassOfCoefficinetsOfDivergence}
with a coupled convex function.   

Jensen-Shannon ($JS$) divergence appeared in the literature relatively
recently~\cite{Lin:1991:DivergenceMeasuresBasedOnTheShannonEntropy}, 
and the unique characteristic of this divergence is that one can
measure a divergence between more than two probability
distributions. Hence, one can term this as a multi-distribution
divergence.

Let $P_1, \ldots , P_M$ be probability distributions and let
$\overline{P} = \sum_{i=1}^M \pi_i P_i $ be a convex combination of
$P_1, \ldots ,P_M$, where $\pi_i \in [0,1]$, $i=1,\ldots,M$ and $\sum_{i=1}^{M} \pi_i = 1$.
Then $JS-$divergence (or, information radius) among $P_1, \ldots , P_M$ is
defined as 
\begin{equation}
\label{JSD}
 JS(P_1, \ldots ,P_M) = \sum_{i=1}^M \pi_i KL(P_i \parallel \overline{P}) \enspace.
\end{equation}

$JS-$divergence is non-negative, symmetric and bounded. For $k=2$, it
has been shown that it is  the square of a
metric~\cite{EndresSchindelin:2003:NewMetricForProbabilityDistributions}.
Grosse et al.\cite{Grosse:2002:AnalysisOfSymbolicSequences} studied several
interpretations and connections of $JS-$divergence. However, the close
association of this divergence to
classification~\cite{Lin:1991:DivergenceMeasuresBasedOnTheShannonEntropy}
makes it quite interesting in our work. 

\subsection*{Related work}
\noindent
Since the seminal work of
\cite{BergerPietraPietra:1996:AmaximumEntropyApproachToNaturalLanguageProcessing}, 
a plethora of maximum entropy techniques for learning have been
introduced in literature. For text classification, 
a simple discriminative maximum entropy technique was used in
\cite{NigamLaffertyMcCallum:1999:UsingMaximumEntropyForTextClassification} 
to estimate the posterior distribution of class variable conditioned
on training data.  
Each feature is assumed to be a function of the data and the class
label. By maximizing the conditional 
entropy of the class labels, we get a distribution of the form 
\begin{equation}
 P(c|x) = \frac{1}{Z(x)}\exp\left(\sum_{i=1}^m{\lambda_i f_i(x,c)}\right),
\end{equation}
where $Z$ is the normalizing constant. The parameters $\lambda_i, 1\le
i\le m$ are obtained by maximizing the conditional 
log likelihood function. It was shown that for some datasets,  
the discriminative maximum entropy model outperformed the more
commonly used multinomial naive Bayes model significantly. An
extension of this approach for sequential data is presented in
~\cite{Mccallum:2000:MaximumEntropyMarkovModels}. 


Now we proceed to the proposed generative maximum entropy
classification and feature selection method.  

\section{MeMd for Binary Classification}
\label{sec_maxEnt_j}
\subsection{Why maximum discrimination?} \label{sec:whymaxdisc}
 What is a `natural' way to select features when the intention is to do maximum entropy Bayes
 classification? In order to answer this question, we first explain what we mean by `natural' with the help of an
 example. For notational convenience, we assume that the classes are labeled as $+1$ and $-1$ in this subsection.
 We will revert to our original notations at the end of this subsection.

 Let us assume that our decision surfaces (or classification
 boundaries) are hyperplanes of the form 
 $f(\mathbf{x}) = \{\mathbf{w}^T \Phi(\mathbf{x}) + b: ||\mathbf{w}||=1 \}$,
 where $\Phi = (\phi_{1},\ldots,\phi_{l})$ is a vector of feature functions. Furthermore, assume that the weight
 vector $\mathbf{w}$ is known and the aim is to select a subset of feature functions. The obvious strategy to
 select features in such a scenario would be to select them so as to maximize either of the two quantities
 \begin{align*}
  \hat{\Phi}_{sum} &= \argmax_{\Phi} \sum_{i=1}^N
  y^{(i)}\left(\mathbf{w}^T \Phi(\mathbf{x}^{(i)}) + b\right) \enspace \\
  \hat{\Phi}_{min} &= \argmax_{\Phi} \min_{i=1}^N y^{(i)}\left(\mathbf{w}^T \Phi(\mathbf{x}^{(i)}) + b\right) ,
 \end{align*}
 where $\mathbf{x}^{(1)}, \mathbf{x}^{(2)} \ldots \mathbf{x}^{(N)}$ are $N$ $i.i.d$ samples that constitute the
 training data and $y^{(i)}$ are their corresponding class labels
 taking values in $\left\{+1, -1\right\}$. Here, the quantity
 $y^{(i)}(\mathbf{w}^T \Phi(\mathbf{x}^{(i)}) + b)$ measures the
 classification
 margin for the point $\mathbf{x}^{(i)}$. Ideally, we would like it to be as high as possible for all points.

 The question that we wish to address is whether we can follow a similar approach to select subset of features
 for the Bayes classifier. In Bayes classification, a point $\mathbf{x}$ is assigned the label $+1$ if
 \begin{equation*}
  \pi_+ P_+(\mathbf{x}) > \pi_- P_-(\mathbf{x}) ,
 \end{equation*}
 which can be rewritten as
 \begin{equation*}
 \log{\frac{\pi_+P_+(\mathbf{x})}{\pi_-P_-(\mathbf{x})}} > 0 .
 \end{equation*}

  Here $\pi_+, \pi_-$ are the prior probabilities and $P_+$ and $P_-$ denote the class conditional
  probabilities of the two classes.
 Hence, the above quantity plays the same role as the equation of
 hyperplane in the former case. Therefore, we can define 
 our Bayes' classification margin for the point $\mathbf{x}$ as
 $y\log{\frac{\pi_+P_+(\mathbf{x})}{\pi_-P_-(\mathbf{x})}}$
 which must be positive for all correctly classified points. As in the case of hyperplanes, we can select features
 so as to maximize either of the two quantities.
 \begin{align}
  {\Gamma}^*_{sum} &=  \underset{S \in 2^{\Gamma}}{\text{arg max}}\sum_{i=1}^N  y^{(i)}\log{\frac{\pi_+P_+(\mathbf{x}^{(i)};S)}{\pi_-P_-(\mathbf{x}^{(i)};S)}} \label{eq:max_sum}\\
  {\Gamma}^*_{min} &=  \underset{S \in 2^{\Gamma}}{\text{arg max}}\min_{i=1}^N y^{(i)}\log{\frac{\pi_+P_+(\mathbf{x}^{(i)};S)}{\pi_-P_-(\mathbf{x}^{(i)};S)}} \notag,
 \end{align}
 where $\Gamma$ is the set of all feature functions and $S \subset \Gamma$.

 If the class conditional distributions are obtained using maximum entropy by using the expected values of feature functions in S,
 we can further simplify \eqref{eq:max_sum} by plugging in the equation for maximum entropy
 distribution for the two classes. The corresponding feature selection problem then becomes
 \begin{align}
  {\Gamma}^*_{sum} &=  \underset{S \in 2^{\Gamma}}{\text{arg max}}\left( \log{\frac{Z_2}{Z_1}} + \log{\frac{\pi_+}{\pi_-} }\right)\left(\sum_{i=1}^N y^{(i)} \right)\nonumber \\ 
   &+\sum_{\phi \in S}(\lambda'_{\phi} - \lambda_{\phi}) \left(
  \sum_{i=1}^N  y^{(i)} \phi(\mathbf{x}^{(i)})
  \right) \label{eq:features},
 \end{align}
 where
 \begin{align}
  P^*_+(\mathbf{x}^{(i)};S)& = \frac{1}{Z_1} \exp\left( -\sum_{\phi \in S} \lambda_{\phi}\phi(\mathbf{x}^{(i)}) \right) \label{eq:class1}\\
  P^*_-(\mathbf{x}^{(i)};S)& = \frac{1}{Z_2} \exp\left( -\sum_{\phi \in S} \lambda'_{\phi}\phi(\mathbf{x}^{(i)}) \right) \label{eq:class2}.
 \end{align}
  We use the superscript `*' to indicate that the class conditional distributions are the ME
 distributions.

 If the number of training points in the two classes are equal, the first term in \eqref{eq:features} can be discarded. 
 Furthermore, by separating the terms in the two classes, we get
 \begin{align}
  {\Gamma}^*_{sum} =  \underset{S \in 2^{\Gamma}}{\text{arg max}}\left(\sum_{\phi \in S} (\lambda_{\phi}' - \lambda_{\phi})(\mu_{\phi} - \mu_{\phi}')\right)\times N \label{eq:j2class}
 \end{align}

 It is easy to see that this is exactly $N$ times the $J$-divergence
 between the  distributions  in \eqref{eq:class1} and
 \eqref{eq:class2}.
 One can similarly show that when the number of points in the two
 classes are not equal, we still obtain the above equation
 if we assign proper weights to instances in the two classes (based on
 the number of points in the class). Reverting back to our
 original notations, we can say that if one intends to use Bayes classifier
 where the class conditional distributions are obtained using maximum
 entropy, a natural way to do so would be to select features as below.
 \begin{equation}
 \Gamma^{*} = \underset{S \in 2^{\Gamma}}{\text{arg max}} \enspace
   J ({ P_{c_{1}}^{*}(x;S) || P_{c_{2}}^{*}(x;S) }),\label{eq:j2problem}\\
 \end{equation}
 where $P_{c_{j}}^{*}(x;S)$ indicates the ME distribution
 estimated for class $c_{j}$ using expected values of the feature
 functions that are in $S \subset \Gamma$.
\subsection{The MeMd Approach}
The MeMd approach can be formulated as follows.
Given a set of feature
functions $\Gamma = \{\phi_{1},\ldots,\phi_{l}\}$, the problem 
is to find the subset $\Gamma^{*} \subset \Gamma$ such
that  
\begin{equation}
\Gamma^{*} = \underset{S \in 2^{\Gamma}}{\text{arg max}} \enspace
  J ({ P_{c_{1}}^{*}(x;S) || P_{c_{2}}^{*}(x;S) }),\\
\end{equation}
where $P_{c_{j}}^{*}(x;S)$ indicates the ME distribution
estimated for class $c_{j}$, $j \in \{1,2\}$, 
using expected values of the feature functions that are in $S \subset \Gamma$.

This problem is intractable particularly when large number
of features are involved (which is the case for high-dimensional data)
since it involves estimation of $J-$divergence for $2^{l}$ subsets
to find the optimal subset from the given set of $l$ feature functions.  
The problem was studied in Dukkipati et al.~\cite{DukkipatiYadavMurty:2010:MaximumEntropyModelBasedClassification} where a greedy search was used to select the features thereby reducing the complexity from exponential to $O(l^{2})$. 
Even with this greedy approach, the method does not scale up well
 for large dimensional data such as text  
data. Moreover  estimating the ME distributions, and hence finding the
exact value of $J-$divergence between  the estimated class conditional
densities is  computationally demanding especially for large
dimensional data.  

Our strategy is to use naive Bayes approach since for text data, naive
Bayes classifiers have shown to outperform (or given comparable
performance) compared with other classifiers~\cite{FriedmanCeigerGoldszmidt:1997:BayesianNetworkClassifiers}
and its good performance is attributed to optimality under zero-one loss
function~\cite{Domingos:1997:OptimalityofBayesClassifier}.  
Therefore we have
\begin{equation*}
P_{c_{j}}(\mathbf{x}) = 
 \prod_{i=1}^{d} P_{c_{j}}^{(i)} ({x}_{i}),
\end{equation*}
where $c_j$ is the class label  
and $P_{c_{j}}^{(i)}$ is the marginal density for $x_{i}$, the
$i^{th}$ feature of the data $\mathbf{x}$.

This leads to simplification of the greedy step 
of the algorithm proposed by Dukkipati et al.~\cite{DukkipatiYadavMurty:2010:MaximumEntropyModelBasedClassification}
to great extent as shown in following result.

\begin{theorem}
\label{lem:J_greedy}
The feature chosen at the $k^{th}$ step of the greedy approach is the one with $k^{th}$ largest $J-$divergence
between the marginals of class conditional densities.
\end{theorem}
\begin{proof}
Using the additivity of $KL-$divergence under independence, $J-$divergence between the two class conditional densities
$P_{c_{1}}$ and $P_{c_{2}}$ can be written as
\begin{align}
J\big(P_{c_{1}} \big\Vert P_{c_{2}} \big) 
= J\left( \prod_{i=1}^{d} P_{c_{1}}^{(i)} \left\Vert \prod_{i=1}^{d} P_{c_{2}}^{(i)} \right. \right) 	
= \sum_{i=1}^{d} J\big(P_{c_{1}}^{(i)} \big\Vert P_{c_{2}}^{(i)} \big)  \enspace.
\label{j_additivity}
\end{align}
Suppose a set $S$ of $(k-1)$ features are already chosen. The corresponding approximation
of the class conditional density is
\begin{displaymath}
  P_{c_j}(\mathbf{x}) \approx \prod_{i\in S} P_{c_{j}}^{(i)}(x_i), \qquad j=1,2,
\end{displaymath}
and the optimal feature is
\begin{align*}
j^{*} &= \underset{j \notin S}{\text{arg max}} \enspace
J\left( \prod_{i\in S\cup\{j\}} P_{c_{1}}^{(i)} \left\Vert \prod_{i\in S\cup\{j\}} P_{c_{2}}^{(i)} \right. \right)
\\ &= \underset{j \notin S}{\text{arg max}} \enspace \left\{
J\left( \prod_{i\in S} P_{c_{1}}^{(i)} \left\Vert \prod_{i\in S} P_{c_{2}}^{(i)} \right. \right)
+ J\left( P_{c_{1}}^{(j)} \left\Vert  P_{c_{2}}^{(j)} \right. \right) \right\}
\\ &= \underset{j \notin S}{\text{arg max}} \enspace 
J\left( P_{c_{1}}^{(j)} \left\Vert  P_{c_{2}}^{(j)} \right. \right) \; .
\end{align*}
The above maximization is equivalent to choosing the feature with $k^{th}$ largest $J-$divergence in the $k^{th}$ step.
\end{proof}

Thus  $J-$divergence can be readily used to rank the features based on the
discrimination between the two classes, and prune out those features
which gives rise to a small value of $J-$divergence between the estimated marginal densities
$P_{c_{j}}^{(i)}(x_{i})$, which are of the ME form
\begin{equation}
 \label{eq:prob}
P_{c_{j}}^{(i)}(x_{i}) = \exp\left(- \lambda_{0_{ij}} - \sum_{k=1}^{l}
\lambda_{k_{ij}}\phi_{k}(x_{i})\right) \enspace,
\end{equation}
for each class $c_j$, $j=1,2$, and each feature $x_i$, $i=1,2,\ldots d$.

For classification, any standard method can be used.
However, since the class conditional densities are estimated during the above process, 
Bayes decision rule~\cite{DudaHartStork:2001:PatternClassification} turns out to be an obvious choice, \textit{i.e.},
a test pattern is assigned class $c_1$ if
\begin{equation*}
 P_{c_1}(\mathbf{x})P(c_{1}) > P_{c_2}(\mathbf{x})P(c_{2}),
\end{equation*}
otherwise to class $c_2$, where $P(c_1)$ and $P(c_2)$ are the priors for each class.
Using only the top $K$ features, the class conditional densities can be approximated as
\begin{equation}
 P_{c_j}(\mathbf{x}) \approx \prod_{i\in S} P_{c_{j}}^{(i)}(x_i), \qquad j=1,2.
 \label{eq:ccd_approx}
\end{equation}
So, the decision rule can be written as
\begin{align*}
\prod_{i\in S} \frac{P_{c_1}^{(i)}(x_i)}{P_{c_2}^{(i)}(x_i)} &>
\frac{P(c_2)}{P(c_1)} \enspace,
\end{align*}
which, after taking logarithm, turns out to be

\begin{align}  
  \sum_{i\in S} ( \lambda_{0_{i2}} - \lambda_{0_{i1}} &+ \sum_{k=1}^{l} (\lambda_{k_{i2}} - \lambda_{k_{i1}} ) \phi_k^{(i)} (x_i) ) \nonumber \\
    &> \ln P(c_2) - \ln P(c_1).
  \label{eq:Bayes_2class}
 \end{align}
We list the above method in Algorithm~\ref{alg:memd_ind}. 
The corresponding experimental results are presented in Section~\ref{sec_exp_results}.

\begin{algorithm}[ht]
\caption{MEMD with Naive Bayes for binary classification}

\textbf{INPUT :}  
\begin{itemize}
 \item 
Two labeled datasets of class $c_{1} $and $c_{2}$.
 \item
Data of the form {$\mathbf{x} = ( x_{1},x_{2}, \ldots ,x_{n} )$}, $x_i$ denoting $i^{th}$ feature. 
 \item
A set of $l$ constraints {$\Gamma^{(i)} = \left\{ \phi_{1}^{(i)},\ldots,\phi_{l}^{(i)} \right\}$} 
to be applied on each feature $x_{i}$, {$i = 1,\ldots,d$}.
\end{itemize}

\textbf{ALGORITHM :}
 \begin{enumerate}
\item
ME densities for each $P_{c_j}^{(i)}$,
$i= 1,2, \ldots,d$ and $j=1,2$  are estimated using~\eqref{eq:prob}.

\item
$J-$divergence for each feature (denote as $J_i$, {$i=1,\ldots,d$}) is calculated using~\eqref{eq:jef}.

\item
The features are ranked in descending order according to their $J-$divergence values, and
the top $K$ features are chosen (to be considered for classification). 

\item
Bayes decision rule is used for classification using~\eqref{eq:Bayes_2class}.
 \end{enumerate}
\label{alg:memd_ind}
\end{algorithm}

An interesting fact to note here is that for distributions of the form in~\eqref{eq:prob}, 
the $J-$divergence can be obtained in a simple form as given in the following result.
\begin{remark}
\label{lem:J_formula}
Suppose there are two ME distributions
$P(x) = \exp\left( -\lambda_{0}  -\sum_{j=1}^{l}\lambda_j  \phi_j(x)\right)$ and
$Q(x) = \exp\left( -\lambda_{0}' -\sum_{j=1}^{l}\lambda_j' \phi_j(x)\right)$,
obtained using same set of feature functions
{$\{\phi_1, \ldots, \phi_l\}$}, but with different expected values
{$\{\mu_1,  \ldots, \mu_l\}$} and {$\{\mu_1',  \ldots, \mu_l'\}$}, respectively.
Then
\begin{equation}
\label{j_addi}
J(P \parallel Q)   = \sum_{j=1}^{l} (\lambda_{j}' - \lambda_{j})(\mu_{j} - \mu_{j}').
\end{equation}
\end{remark}

This result can be used to evaluate the $J-$divergence between the
marginals of the class conditional 
densities for each feature, which can be used to rank all the features in $O(d)$ time 
(in decreasing value of $J-$divergence). 
\section{Multi Class Classification}
\label{sec_multiclass}
\subsection{One vs. All Approach}
\label{sec_onevsall}

The most common technique in practice for multi-class
classification is to use ``\emph{one vs. all}'' approach, where the number of
classifiers built is equal to the number of classes, \textit{i.e.}, for each class we build a
classifier for that class against all the other
classes~\cite{AllweinShapireSinger:2000:ReducingMulticlassToBinary}.     
Incorporating such a technique does not affect the basic MeMd approach for ranking the features.
Hence, Algorithm~\ref{alg:memd_ind} can be easily extended
to multi-class case by using one vs. all
approach as presented in Algorithm~\ref{alg:onevsall_multi}. 

Here, we
consider a $M$-class problem, with classes {$c_1, c_2, \ldots, c_M$}.
The rest of the setting is same as Algorithm~\ref{alg:onevsall_multi}.
The modification can be described as follows. For each class $c_j$,
consider the class {$c_j' = \cup_{k\neq j} c_k$}. 
So, the `discriminative capacity' of 
each feature for a particular class can be measured by its $J-$divergence
between the class conditional densities $c_j$ and $c_j'$,
\begin{align}
 J_{ij} &:= J \left( P_{c_j}^{(i)} \left\Vert P_{c_j'}^{(i)}
 \right. \right) \enspace,
 \label{eq:J_1vsAll}
\end{align}
where $P_{c_j}^{(i)}$ and $P_{c_j'}^{(i)}$ are the ME marginal densities
for classes $c_j$ and $c_j'$, respectively, for the {$i^{th}$} feature.
$J_{ij}$ can be easily computed using~\eqref{j_addi}. However, this provides us a 
$J-$divergence for each feature for a particular class. A natural way to 
obtain a $J-$divergence for each feature is 
\begin{equation}
 J_i = \sum\limits_{j=1}^{M} J_{ij} P(c_j),
 \label{eq:jef_avg}
\end{equation}
\textit{i.e.}, the divergence is averaged over all the classes, weighted 
by their prior probabilities {$P(c_j)$}.
The algorithm is listed as below.

\begin{algorithm}[ht]
\caption{(MeMd-J) : MEMD for  multi-class classification using one vs. all approach.}

\textbf{INPUT :}  
\begin{itemize}
 \item 
Labeled datasets of $M$ classes.
 \item
Data of the form {$\mathbf{x} = ( x_{1},x_{2}, \ldots ,x_{n} )$}, $x_i$ denoting $i^{th}$ feature. 
 \item
A set of $l$ constraints {$\Gamma^{(i)} = \left\{ \phi_{1}^{(i)},\ldots,\phi_{l}^{(i)} \right\}$} 
to be applied on each feature $x_{i}$, {$i = 1,\ldots,d$}.
\end{itemize}

 \textbf{ALGORITHM :}
 \begin{enumerate}
\item
Construct the classes {$c_j' = \bigcup\limits_{k\neq j} c_k$}, {$j = 1,2,\ldots, M$}.
 
\item
ME densities $P_{c_j}^{(i)}$ and $P_{c_j'}^{(i)}$,
$i= 1,2, \ldots,d$ and $j=1,2,\ldots, M$  are estimated using~\eqref{eq:prob}.

\item
The average $J-$divergence for each feature is calculated using~\eqref{eq:jef_avg}.

\item
The features are ranked in descending order according to their average $J-$divergence values, and
the top $K$ features are chosen (to be considered for classification). 

\item
Bayes decision rule is used to assign a test
pattern to $c_{j^{*}}$ such that
\begin{displaymath}
 j^{*} = \underset{j = 1,\dots, M}{\text{arg max}}  \enspace P_{c_j}(\mathbf{x})P(c_{j}),
\end{displaymath}
where the class conditional densities are approximated as in~\eqref{eq:ccd_approx} for {$j = 1,2,\ldots,M$}.

\end{enumerate}

\label{alg:onevsall_multi}
\end{algorithm}

The above approach is an obvious extension of the binary classification problem.
Although the algorithm is $O(d)$, it
is relatively computationally inefficient in the sense that it requires estimation
of the additional ME distributions $P_{c_j'}^{(i)}$ 
(which is the most time consuming step of MeMd approach).
Now, we present our main contribution in this paper. To deal with
the multi-class classification, we invoke a natural multi-distribution divergence,
the $JS-$divergence.
\subsection{Multi class Classification using $JS-$divergence}
\label{sec_multiclass_js}
The $JS-$divergence, as mentioned in Section~\ref{sec_preliminaries}
can be defined over multiple distributions, and hence, would seem to
be more useful for dealing with multi-class problems. However, its
most interesting feature is presented in the following
result~\cite{Grosse:2002:AnalysisOfSymbolicSequences}. We state it in
a form more suitable for our purpose, and interpret it for
multi-class classification similarity. 
\begin{proposition}
Consider a $M$-class problem with class-conditional densities
$P_{c_{i}}(x) = P(x|c_i)$, $i=1,2,\ldots,M$. If $X$ is a random data,
and $Z$ is an indicator random variable, \textit{i.e.}, $Z=i$ when
$X\in c_i$, then 
\begin{displaymath}
 JS(P_{c_{1}}, \ldots ,P_{c_{M}}) = MI(X;Z).
\end{displaymath}
Here, $MI(X;Z)$ is the mutual information between $X$ and $Z$,
and the $JS-$divergence is computed using the priors as weights,
(\textit{i.e.}, {$\pi_j = P(c_j)$}) as 
\begin{equation}
 JS(P_{c_{1}}, \ldots ,P_{c_{M}}) = \sum_{j=1}^{M} P(c_j) KL( P_{c_j} \Vert \overline{P}),
 \label{eq:js_compute}
\end{equation}
where {$\overline{P}(X) = \sum\limits_{j=1}^{M} P(c_j) P_{c_j}(X)$} is the probability 
of sampling the data $X$.
\end{proposition}

The problem of classification basically deals prediction of the value of $Z$,
given the test data $X$. Hence, it is natural to device a technique which
would maximize the mutual information, or rather, maximize the $JS-$divergence
among the class-conditional densities. Thus $JS-$divergence turns out to be a more
natural measure of discrimination compared to $J-$divergence.
However, unlike $J-$divergence, $JS-$divergence does not satisfy 
the additivity property, \textit{i.e.}, there is no equivalent
for~\eqref{j_additivity} for $JS-$divergence. 
Hence, the ``\textit{correct}'' extension of the algorithm in
case of $JS-$divergence becomes involved as we need to consider 
a quadratic time greedy approach,
where at each iteration we select a new feature that maximizes 
$JS-$divergence among the classes and also has minimum mutual information
with the previously chosen feature set.

The latter term leads to an $O(d^2)$ time complexity for the 
feature selection rule. 
A linear time algorithm similar to MeMd-J can be obtained if we assume 
that all the features are independent, which is a quite strict assumption,
particularly for text datasets.
In spite of the above issues, use of $JS-$divergence is quite significant
due to its theoretical justifications. Moreover, use of a multi-distribution
divergence significantly reduces the the number of models that needs to
be estimated in a one-vs-all approach.
Hence, we look for an approximation of the $JS-$divergence that  
some-what retains its properties, and can also exploit the linearity of $J-$divergence.
\section{$JS_{GM}-$ discrimination}
\label{sec_newdiv}
\subsection{AM to GM}
Although $JS-$divergence seemed to be a natural discriminative measure, its drawbacks from
algorithmic point of view motivates us to look into some form of approximation in terms of
$J-$divergence, which looked more promising (computationally).
The idea comes from~\cite{Sgarro:1979:InformationDivergence}, which introduced the notion of 
average divergence by just taking average over $KL-$divergence between all possible pairs
of distributions. 

However, we approach the problem differently, but obtain a similar result.
The modification is made at the basic definition of $JS-$divergence~\eqref{JSD},
where we replace the weighted mean $\overline{P} = \sum_{i=1}^k \pi_i P_i $
by the weighted geometric mean $\overline{P}' = \prod_{i=1}^k P_i^{\pi_i} $.
Though $\overline{P}'$ does not have any straightforward physical interpretation
(in fact, it is not a distribution),
it leads to simpler expressions as discussed below. We call the corresponding
divergence as $JS-$divergence with geometric mean, or simply $JS_{GM}-$divergence.
Given distributions {$P_1, P_2, \ldots, P_M$} with weights {$\pi_1, \pi_2, \ldots, \pi_M$},
the $JS_{GM}-$divergence among them is given by
\begin{align}
 JS_{GM} (P_1,  \ldots, P_M) &= \sum_{i=1}^{M} \pi_i KL(P_i \Vert \overline{P}') 			\nonumber
 \\&=  \sum_{i=1}^{M} \sum_{j\neq i} \pi_i \pi_j KL(P_i \Vert P_j)	\; .	
 \label{eq:JSgm_KLform}
\end{align}
In the case of uniform weights ({$\pi_i = \frac{1}{M}$}), it is same as the average divergence~\cite{Sgarro:1979:InformationDivergence}
upto a constant scaling factor. Observing the symmetric nature of~\eqref{eq:JSgm_KLform},
the divergence can also be written in terms of $J-$divergence as
\begin{equation}
 JS_{GM} (P_1,  \ldots, P_M) =  \frac{1}{2} \sum_{i=1}^{M} \sum_{j\neq i} \pi_i \pi_j J(P_i \Vert P_j) \;.	
 \label{eq:JSgm_Jform}
\end{equation}

The $JS_{GM}-$divergence satisfies the following inequality, which has been shown
in~\cite{Sgarro:1979:InformationDivergence} for uniform weights.
\begin{proposition}
 \label{lem:JS_JSgm}
 For any set of weights {$\pi_1, \pi_2, \ldots, \pi_M$},
 \begin{displaymath}
  JS(P_1,  \ldots, P_M) \leqslant JS_{GM}(P_1, \ldots, P_M)
 \end{displaymath}
\end{proposition}
\begin{proof}
 The claim follows from the observation that $KL-$divergence is convex in its second argument.
\end{proof}

The usefulness of the $JS_{GM}-$divergence stems from the fact that
it is a multi-distribution divergence that generalizes $J-$divergence, 
which helps us to overcome the difficulties of the one vs. all approach. So it does not
requires estimation of extra ME distributions. At the same time, it can be expressed
in terms of $J-$divergence~\eqref{eq:JSgm_Jform}, which helps us to exploit the nice properties 
of $J-$divergence discussed in Theorem~\ref{lem:J_greedy} and Remark~\ref{lem:J_formula}.
So, 
the following algorithm using $JS_{GM}-$divergence is analogous to the binary 
classification algorithm with $J-$divergence (Algorithm~\ref{alg:memd_ind}).
The equivalence can be easily observed by considering the following equivalence.
For each feature, we replace

\begin{align}
 J &\left( P_{c_j}^{(i)} \left\Vert P_{c_j'}^{(i)} \right. \right)
 \text{ in Algo~\ref{alg:onevsall_multi} }
 \longleftrightarrow \nonumber \\ 
 &\sum_{k\neq j} P(c_k) J \left( P_{c_j}^{(i)} \left\Vert P_{c_k}^{(i)} \right. \right) 
 \text{ in Algo~\ref{alg:JSGM}},
\end{align}
\textit{i.e.}, $J-$divergence between class conditional densities for $c_j$ and all other 
classes taken together ($c_{j'}$) is replaced by the weighted average of the $J-$divergences
between $c_j$ and other classes.

\begin{algorithm}
 \caption{(MeMd-JS) : MEMD for multi-class classification using $JS_{GM}-$divergence}
 
\textbf{INPUT :}  
\begin{itemize}
 \item 
Labeled datasets of $M$ classes.
 \item
Data of the form {$\mathbf{x} = ( x_{1},x_{2}, \ldots ,x_{n} )$}, $x_i$ denoting $i^{th}$ feature. 
 \item
A set of $l$ constraints {$\Gamma^{(i)} = \left\{ \phi_{1}^{(i)},\ldots,\phi_{l}^{(i)} \right\}$} 
to be applied on each feature $x_{i}$, {$i = 1,\ldots,d$}.
\end{itemize}

 \textbf{ALGORITHM :}  
 
  \begin{enumerate}
\item
ME densities $P_{c_j}^{(i)}$,
$i= 1,2, \ldots,d$ and $j=1,2,\ldots, M$  are estimated using~\eqref{eq:prob}.

\item
The $J-$divergence for each feature and each pair of classes is calculated using~\eqref{eq:js_compute}.

\item
The $JS_{GM}-$divergence for each feature can be computed as

\begin{align*}
 JS_{GM} ^{(i)} &= JS_{GM} (P_{c_1}^{(i)}, P_{c_2}^{(i)}, \ldots, P_{c_M}^{(i)}) \nonumber
  \\&=  \sum_{j=1}^{M} \sum_{k\neq j} P(c_j) P(c_k) J \left( P_{c_j}^{(i)} \left\Vert P_{c_k}^{(i)} \right. \right)	\; .	
\end{align*}

\item
The features are ranked in descending order according to their
$JS_{GM}-$divergence values, and 
the top $K$ features are chosen (to be considered for classification). 

\item
Bayes decision rule is used to assign a test
pattern to $c_{j^{*}}$ such that
\begin{displaymath}
 j^{*} = \underset{j = 1,\dots, M}{\text{arg max}}  \enspace P_{c_j}(\mathbf{x})P(c_{j}),
\end{displaymath}
where the class conditional densities are approximated as in~\eqref{eq:ccd_approx} for {$j = 1,2,\ldots,M$}.

\end{enumerate}
\label{alg:JSGM}
\end{algorithm}

We can also extend the arguments in Section~\ref{sec:whymaxdisc} to show that for maximum entropy Bayes classification in the multi-class scenario, $JS_{GM}-$divergence is a natural choice. The algorithm has linear time complexity $O(d)$, and requires half the
time as compared to  Algorithm~\ref{alg:onevsall_multi} (one vs. all).
However, the performance of both algorithms are quite similar as
discussed in the experimental results.

\begin{table*}[ht]
 \centering
 \caption{Comparison of complexity of algorithms.}
 \begin{tabular}{|c|c|c|c|c|c|}
  \hline
  Algorithm	& Notation	& Classification & \multicolumn{2}{c|}{Training time}		& Testing time 	\\
  \cline{4-5}
		&&	& Estimation	& Feature ranking	& per sample	\\
  \hline \hline
  
  MeMd using one vs. all approach & MeMd-J	& multiclass	& O($MNd$)		& O($Md +d\log d$)	& O($MK$)	\\
  MeMd using $JS_{GM}-$divergence & MeMd-JS	& multiclass	& O($MNd$)		& O($M^2d+d\log d$)		& O($MK$)	\\
  \cline{4-5}
  MeMd using greedy approach~\cite{DukkipatiYadavMurty:2010:MaximumEntropyModelBasedClassification}
  & MeMd	& binary	& \multicolumn{2}{c|}{O($Nd^2$)}		& O($Md$)	\\
  Support Vector Machine~\cite{ChangLin:2011:LIBSVM}	& SVM 	& multiclass	& \multicolumn{2}{c|}{\#iterations*O($Md$)}  	& O($M^2Sd$)	\\
  Discriminative approach using ME~\cite{NigamLaffertyMcCallum:1999:UsingMaximumEntropyForTextClassification} & DME	& multiclass	& \multicolumn{2}{c|}{\#iterations*O($MNd$)}	& O($Md$)	\\
  \hline
 \end{tabular}
 \label{tab:algo_complexity}
\end{table*}

\begin{table*}
 \centering 
 \caption{Performance comparison on gene expression datasets.}
 \label{tab:gene}
 \begin{tabular}{|c||c|c|c||c|c|c|}
  \hline 
	      & \multicolumn{3}{c||}{Data attributes}	& \multicolumn{3}{c|}{10-fold cross validation accuracy} \\
  \cline{2-7}
    Dataset 	& No. of  & No. of  & No. of		& SVM      & MeMD-J     & MeMd-JS	\\
		& classes & samples & features		& (linear) & (2-moment) & (2-moment)	\\
  \hline \hline
  Colon cancer~\cite{ARFF:Datasets}
		& 2	  & 62	    & 2000		& 84.00	   & \multicolumn{2}{c|}{\bf 86.40}	\\
  Leukemia~\cite{ARFF:Datasets,MramorLeban:2007:VisualizationBasedCancerMicroarray}
		& 2	  & 72	    & 5147		& 96.89	   & \multicolumn{2}{c|}{\bf 98.97}	\\
  Embryonal tumors of CNS~\cite{ARFF:Datasets}
		& 2	  & 60	    & 7129		& 62.50	   & \multicolumn{2}{c|}{\bf 63.75}	\\
  DLBCL~\cite{MramorLeban:2007:VisualizationBasedCancerMicroarray}
		& 2	  & 77	    & 7070		& {\bf 97.74}   & \multicolumn{2}{c|}{86.77}	\\
  Prostate cancer~\cite{MramorLeban:2007:VisualizationBasedCancerMicroarray}
		& 2	  & 102	    & 12533		& 89.51	   & \multicolumn{2}{c|}{\bf 89.75}	\\
  \cline{6-7}
  SRBCT~\cite{MramorLeban:2007:VisualizationBasedCancerMicroarray}
		& 4	  & 83	    & 2308		& {\bf 99.20}   &     97.27       & 98.33	\\
  Human lung carcinomas~\cite{MramorLeban:2007:VisualizationBasedCancerMicroarray}
		& 5	  & 203	    & 12600		& 93.21	   &  {\bf 93.52}          & 92.60	\\
  Global Cancer Map~\cite{ARFF:Datasets}
		& 14	  & 190	    & 16063		& 66.85	   &  {\bf 66.98}          & {\bf 66.98}	\\
  \hline 
 \end{tabular}
\end{table*}

\section{Experimental Results}
\label{sec_exp_results}

In this section, we compare the proposed algorithms with some popular
algorithms,  commonly used in practice. Table~\ref{tab:algo_complexity}
provides a summary of the proposed algorithms, as well as the existing
algorithms. We also introduce some notations for each algorithm, which
will be used in sequel to refer to the algorithm. We note that the
binary classification algorithm (Algorithm~\ref{alg:memd_ind}) is a
special case of the multi-class \emph{``one vs. all''} approach
(Algorithm~\ref{alg:onevsall_multi}),  and so, in sequel, we always
refer to both together as MeMd-J algorithm. 
\subsection{Summary of computational complexity of algorithms}
\label{sec_complexity}
We first discuss about the computational complexity of the various
algorithms. This is detailed in Table~\ref{tab:algo_complexity}.   It
shows the training and testing time complexities of the various
algorithms in terms of number of features ($d$), number of classes
($M$) and  number of training samples ($N$). For SVM, $S$ denotes the number of
support vectors, while for proposed algorithms $K$ is the chosen number of 
features (obtained from the ranking). The training time complexity  for SVM assumes that
the rows have been cached, as in the case of LIBSVM. 
SVM was implemented using
LIBSVM~\cite{ChangLin:2011:LIBSVM}, which uses a \emph{``one
  vs. one''} approach. Furthermore,  the number of iterations for SVM
tend to be $O(N^2)$, whereas for DME, the algorithm gives best results
when the number of iterations is small ($O(1)$) as mentioned
in~\cite{NigamLaffertyMcCallum:1999:UsingMaximumEntropyForTextClassification}. 

The greedy MeMd
technique~\cite{DukkipatiYadavMurty:2010:MaximumEntropyModelBasedClassification} 
is severely affected by curse of dimensionality. In addition to its
quadratic complexity, implementations indicated that the 2-moment ME
joint distribution (Gaussian) appeared to be unstable due to the
inverse of the large dimensional covariance matrix. Hence, results
corresponding to this algorithm have not been presented in the
comparisons. We also skip the MeMd algorithm using $JS-$divergence
since determination of true value of $JS-$divergence involves
considerable amount of computation as it requires numerical
calculation of integrals. Hence, we work with its approximate version
(Algorithm~\ref{alg:JSGM}).  It is worth noting here that though
MeMd-JS and MeMd-J have same order of complexity for parameter
estimation, for each class, MeMd-J builds a model  for all data not in
the class. Hence, it estimates twice the number of parameters. Hence,
using MeMd-JS over MeMd-J is computationally efficient in the
following cases. 
\begin{enumerate}
 \item 
 The number of points is large thereby making parameter estimation time for MeMd-J twice as MeMd-JS. 
 Note that feature selection phase does not depend on the number of points, hence, having an $M^2$ term in MeMd-JS is not a problem.
 \item
 The number of iterations required for parameter estimation is not O(1),
 \textit{i.e.}, when we use iterative scaling algorithm. 
 However, this issue is not there if only 1-moment or 2-moment models are considered. 
\end{enumerate}
\subsection{Experiments on biological datasets}
We compare the performance of our algorithms (MeMd-J and MeMd-JS)
with that of SVM on a variety of gene datasets, collected
from~\cite{ARFF:Datasets,MramorLeban:2007:VisualizationBasedCancerMicroarray}.
The details of the datasets are presented in Table~\ref{tab:gene},
where the number of features is same as the number of genes in these datasets.
We again mention here that the greedy
MeMd~\cite{DukkipatiYadavMurty:2010:MaximumEntropyModelBasedClassification} 
and MeMd using $JS-$divergence have not been implemented due to their 
computational complexities addressed in the previous section.

Table~\ref{tab:gene} also lists the accuracies obtained from different algorithms using 10-fold cross-validation,
where the folds are chosen randomly. The class conditional distributions are 2-moment ME distributions.
 The optimal number of features is selected by cross validation within the training sample.
 We note here that for binary classification,
the procedure in both MeMd-J and MeMd-JS are same, \textit{i.e.}, they result in the same algorithm.
Hence, in these cases, the results for both algorithms are combined together.

The best accuracy for each dataset is highlighted. 
From Table~\ref{tab:gene}, one can observe that the MeMd classifier is more successful in distinguishing among the
classes for most of the cases. This is because, a very small number of genes are actually involved in these diseases.
By using a feature pruning strategy in the classification algorithm, the MeMd classifier is able to prune away most of the
genes that are not discriminatory. 

\begin{table*}
 \centering 
 \caption{Performance comparison on 20 Newsgroup dataset ('/' used to separate classes; names abbreviated in some cases).}
 \label{tab:classify_20news}
 \begin{tabular}{|c||c|c|c||c|c|c|c|}
  \hline 
	      & \multicolumn{3}{c||}{Data attributes}	& \multicolumn{4}{c|}{2-fold cross validation accuracy} \\
  \cline{2-8}
    Experiment 	& No. of  & No. of  & No. of		& SVM      & DME	& MeMD-J     & MeMd-JS		\\
		& classes & samples & features		& (linear) &		& (1-moment) & (1-moment)	\\
  \hline \hline
  rec.autos / rec.motorcycles
		& 2	  & 1588    & 7777		& 95.96	   & 95.59	& \multicolumn{2}{c|}{\bf 97.35}	\\
  talk.religion / soc.religion
		& 2	  & 1227    & 8776		& 91.35	   & {\bf 92.33}& \multicolumn{2}{c|}{91.69}	\\
  comp.os.ms-windows / comp.windows.x
		& 2	  & 1973    & 9939		& 92.80	   & 93.60	& \multicolumn{2}{c|}{\bf 93.81}	\\
  comp.sys.ibm / comp.sys.mac
		& 2	  & 1945    & 6970		& 89.61    & {\bf 90.48}& \multicolumn{2}{c|}{89.77}	\\
  rec.(autos,motorcycles) / rec.sport
		& 2	  & 3581    & 13824		& 98.49	   & 98.68	& \multicolumn{2}{c|}{\bf 99.02}	\\
  sci.(space,crypt) / sci.(med,elec)
		& 2	  & 3952    & 17277		& 96.63	   & {\bf 96.93}& \multicolumn{2}{c|}{95.04}	\\
  comp.windows(os,x) / comp.sys.(ibm,mac)
		& 2	  & 3918    & 13306		& 88.23	   & {\bf 91.88}& \multicolumn{2}{c|}{\underline{91.75}}	\\
  \cline{7-8}
  all `talk' categories
 		& 4	  & 3253    & 17998		& 88.62    & 90.34    	& {\bf 91.91}  & 91.39	\\
  all `sci' categories
 		& 4	  & 3952    & 17275		& 94.63    & {\bf 95.26}   & \underline{95.14}   & 94.88	\\
  all `rec' categories
 		& 4	  & 3581    & 13822		& 95.83    & {\bf 96.23}   & \underline{96.14}   & 95.86	\\
  all `comp' categories
 		& 4	  & 4891    & 15929		& 81.08    & {\bf 83.41}   & 82.11   & 82.03	\\
  \hline 
 \end{tabular}
\end{table*}

\subsection{Experiments on text datasets}
\label{sec_exp_setup}
We perform our experiments on 
the 20 Newsgroups dataset obtained from~\cite{BacheLichman:2013:UCIMachineLearning}.
The dataset contains 20  classes, which have been grouped in different ways
to construct different binary and multiclass problems.
The data has been preprocessed prior to classification.
We remove all stop words, and words with frequency less than some cut-off value $(\gamma)$ from the 
entire document corpus. 
This is done since such words do not have much discriminative power.
The value of $\gamma$ is chosen to be 2 in all cases.
Each document $D_{j}$ of the text corpus is represented as a vector of 
term-weights $D_{j} = \langle W_{1j},W_{2j}, \ldots ,W_{Tj}\rangle $,  where $W_{ij}$
represents the normalized frequency of the word $w_{i}$ in the
document $D_{j}$~\cite{Sebastiani:2002:MachineLearningInAutomatedTextCategorization},
\textit{i.e.},
$ W_{ij} =\frac{N(w_{i},D_{j})}{\sum_{k = 1}^{T}{N(w_{k},D_{j}})}$,
where $N(w_{i},d_{j})$ is the number of times the word $w_{i}$ occurs in the document $D_{j}$,
and $T$ is the total number of words in the corpus. 
Thus $0 \leq W_{kj} \leq 1$ represents how much the word $w_{k}$ contributes to the semantics of the
document $D_{j}$. 

Table~\ref{tab:classify_20news}  presents the classification accuracy of 
the proposed methods (MeMd-J and MeMd-JS) along with SVM and DME using 2-fold cross-validation.
The main reason behind using only two folds is that, in such a case,
the training data size reduces considerably, and learning correctly from small data 
becomes an additional challenge. We demonstrate how the proposed generative approach
overcomes this problem.
Furthermore, MeMd-J and MEMd-JS algorithms use only the top $K$ ranked features 
for classification. For choosing the the optimal value of $K$, we employ the following strategy. 
We divide the training data 
in two portions: a considerable fraction (80\%) of the training data is used to rank
the features, while the remaining 20\% is used as a ``test'' set on 
which classification is performed using varying number of features. Hence, we obtain 
a plot of the accuracy vs. the number of features. We choose $K$ to be the 
minimum number of features, where the maximum accuracy is attained for that small portion of training data.
We also note, as in previous section, that for binary classification MeMd-J and MeMd-JS
are exactly same, and so we present their results together.

The best accuracies for each experiment is shown in bold. 
We also underline the cases, where results are very close to best case (more than -0.2\% of maximum accuracy).
Our observation can be summarized as follows:
In all of the cases (except one), MeMd outperforms linear SVM, which showed quite poor performance. 
We also implemented SVM with polynomial and RBF, which exhibited even poorer performance, 
and hence, those results are not presented.
On the other hand, we observed that DME performs quite well, 
providing the best accuracies in a considerable number of cases.
Further, in multiclass problems, MeMd-J is always observed to perform marginally better than MeMd-JS.
\section{Conclusion}
\label{conclusion}
As far as our knowledge, this is the first work that proposes
and studies a generative maximum entropy approach to classification.
In this paper, we proposed a method of classification using maximum
entropy with maximum discrimination (MeMd) which is a generative approach
with simultaneous feature selection.  
The proposed method is suitable for large dimensional text dataset as
the classifier is built in linear time with 
respect to the number of features and it provides a way to eliminate
redundant features. 
Also this is the first work that uses multi-distribution divergence in 
multiclass classification. It will also be interesting to study the proposed methods as feature selection algorithms.
\bibliographystyle{IEEEtran}
\bibliography{memdCameraConf}
\end{document}